\documentclass[aps, reprint, amsmath,amssymb, pra, superscriptaddress]{revtex4-2}

\usepackage[title]{appendix}
\usepackage{color}
\usepackage{graphicx}
\usepackage{dcolumn}
\usepackage{bm}
\usepackage{physics}
\usepackage{hyperref}
\usepackage{url}
\usepackage{subcaption}
\usepackage{amsmath}
\usepackage{notes2bib}
\usepackage{ dsfont }
\usepackage{verbatim}
\usepackage{amssymb}
\usepackage{amsfonts}              
\usepackage{amsthm}                
\usepackage{mathtools}
\usepackage{kotex}
\usepackage{qcircuit}
\usepackage[normalem]{ulem}
\usepackage{xcolor}

\newtheorem{theorem}{Theorem}

\newtheorem{definition}[theorem]{Definition}
\newtheorem{corollary}[theorem]{Corollary}

\bibnotesetup{note-name    = ,use-sort-key = false} 
\begin{document}

\title{The Min-entropy as a Resource for One-Shot\\ Private State Transfer, Quantum Masking and State Transition}

\author{Seok Hyung Lie}
\author{Seongjeon Choi}
\author{Hyunseok Jeong}
\affiliation{%
 Department of Physics and Astronomy, Seoul National University, Seoul, 151-742, Korea
}%

\date{\today}

\begin{abstract}
We give an operational meaning to the min-entropy of a quantum state as a resource measure for various interconnected tasks. In particular, we show that the min-entropy without smoothing measures the amount of quantum information that can be hidden or encoded perfectly in the one-shot setting when the quantum state is used as a randomness/correlation source. First, we show that the min-entropy of entanglement of a pure bipartite state is the maximum number of qubits privately transferable when the state is used as quantum one-time pad. Then, through the equivalence of quantum secret sharing(QSS)-like protocols, it is also shown that the min-entropy of a quantum state is the maximum number of qubits that can be masked when the state is used as a randomness source for a quantum masking process. Consequently we show that the min-entropy of a quantum state is the half of the size of quantum state it can catalytically dephase.This gives a necessary and sufficient condition for catalysts for state transition processes.
\end{abstract}

\pacs{Valid PACS appear here}
\maketitle

\section{Introduction}
The space of quantum correlation is very vast. The dimension of a collection of many quantum systems is much larger than the sum of the dimension of each system. It has motivated researches on the method of encoding information within global quantum state without altering local quantum systems. Such efforts have appeared under many names; quantum error correcting codes \cite{steane1996multiple, lidar2013quantum}, quantum secret sharing\cite{gottesman2000theory, cleve1999share}, quantum masking \cite{modi2018masking, li2020probabilistic,li2018masking} and private state transfer\cite{schumacher2006quantum, oppenheim2012quantum}.

Among these tasks, quantum secret sharing (QSS) \cite{cleve1999share} is especially important since, as we will see, it subsumes many other similar tasks. QSS is the task of distributing an arbitrary quantum state to multiple parties in a fashion that only authorized subsets of them can restore the quantum state. Each local party's marginal state (`share') of a QSS scheme should have a constant form regardless of the quantum secret. 
Typically each share of a QSS scheme is a quantum state that should be stored in a quantum system. However, it is still demanding to maintain a large size of quantum system protected from noise and error. Thus, estimating and optimizing the informational size of the each share is critical, since it is directly related to the required physical size of the storage medium to contain each share.

The informational size of a quantum system is decided by how random the system is. There have been studies on lower bounds of the amount of randomness of each share in QSS scheme. The R\'{e}nyi entropy, defined as
\begin{equation}
    S_\alpha(\rho)\equiv \frac{1}{1-\alpha}\log_2\Tr[\rho^\alpha],
\end{equation}
and their limits, the max-entropy $S_{\max{}}(\rho) \equiv \log_2\text{rank}(\rho)= \lim_{\alpha\to0}S_\alpha(\rho)$, the min-entropy $S_{\min{}}(\rho) \equiv -\log_2\|\rho\| = \lim_{\alpha\to\infty} S_\alpha (\rho)$ and the von Neumann entropy $S(\rho)\equiv -\Tr\left[\rho \log_2 \rho \right] =\lim_{\alpha\to1} S_\alpha(\rho)$ are often used to quantify the randomness within a quantum state $\rho$. In Ref. \cite{gottesman2000theory}, it was proven that, for arbitrary secret sharing scheme for $d-$dimensional quantum secret, the dimension of each share of secret must be at least as large as the dimension of the secret itself. This provides a lower bound for the max-entropy of each share, i.e. $S_{\max{}} (\sigma) \geq \log_2 d$. In Ref. \cite{imai2005information, lie2019unconditionally}, the result was improved to provide a lower bound of the von Neumann entropy of each share, i.e. $S(\sigma) \geq \log_2 d$. Note that the R\'{e}nyi entropy monotonically decreases as $\alpha$ grows \cite{hayashi2017quantum}.

The problem, however, was not closed, since the optimality of the lower bound was not proved. Can any quantum state with the von Neumann entanglement entropy larger than $\log_2 d$ be a marginal state of a QSS scheme? If not, when is it possible?

In this work, we show that this question is intimately related to other questions about the amount of required resources for many other important quantum information processing tasks. We then close this problem by giving the \textit{min-entropy} of quantum state operational meanings as the power for tasks such as private state transfer, quantum masking and implementation of dephasing map.

For example, for the quantum masking \cite{modi2018masking, lie2019unconditionally, lie2020randomness}, the task of hiding quantum information into bipartite quantum correlation using randomness source and bipartite interaction, the amount quantum information that can be masked by a randomness source is given by its min-entropy. For the private state transfer \cite{schumacher2006quantum}, the task of transmitting a quantum state without giving any information to a potential eavesdropper by utilizing pre-established quantum correlation, the amount of privately transferable quantum information is determined by the min-entropy of the marginal state of the pre-established pure bipartite state.

In doing so, we introduce a \textit{deterministic} method of randomness extraction from a weak quantum randomness source, i.e. a mixed state with high enough min-entropy but having non-uniform eigenvalues, utilizing the Nielsen theorem \cite{nielsen1999conditions}.

These results imply an important consequence for state transition processes under the constraint that randomness is not free, which is deeply related to quantum thermodynamics \cite{gour2015resource, muller2018correlating}. We completely characterize the randomness sources that can dephase a given size of quantum system. As a direct consequence, we drive a necessary and sufficient criterion for the possibility of state transition process with a given \textit{catalyst} \cite{boes2018catalytic}.

\section{Private state transfer} 

First, we give the definition of encoding schemes for faithful one-shot private state transfer (PST). Consider a situation in which two parties, Alice and Bob, has pre-distributed entangled state $\ket{\Psi}_{AB}$. Alice encodes her possibly unknown quantum state $\psi$ by making $\psi$ interact with her part of $\ket{\Psi}_{AB}$. It results in the secret encoding channel $\Phi_\psi$ acting on the system $A$ of $\ket{\Psi}_{AB}$. Then Alice transmits the system $A$ to Bob over a quantum channel. However, to make the secret remain private, any possible eavesdropper seizing the transmitted state $\Phi_\psi(\Tr_B \dyad{\Psi}_{AB})$ should gain no information at all about the state $\psi$. To finish the transmission, there also should be a recovery map that can recover $\psi$ from $(\Phi_\psi \otimes \mathcal{I})(\dyad{\Psi}_{AB})$. We will focus on the case where this recovery map exactly recovers the secret state, in contrast to \textit{approximate} recovery. Now we give the technical description of this task.

We will denote the Hilbert space corresponding to quantum system $A$ as $\mathcal{H}_A$ and the vector space of operators on the Hilbert space $\mathcal{H}_A$ as $\mathcal{B(H}_A)$. We will also follow the convention of denoting the marginal state on the system $A$ of a multipartite state $\ket{\Psi}_{ABC\dots}$ as $\Psi_A$ throughout this work. In the following definition,  families of quantum channels defined on $\mathcal{B}(\mathcal{H}_A)$ with the form $\{\mathcal{E}_\psi\}$ will be considered, where the index $\psi$ can be arbitrary $d$-dimensional quantum state. 

\begin{definition}[Private state transfer]
A family of quantum channels $\{\mathcal{E}_\psi\}$ is said to encode quantum state $\psi$ into a bipartite state $\ket{\Psi}_{AB}$ for $d$-dimensional faithful one-shot PST if $\Phi_\psi(\Psi_A)$ is constant regardless of $\psi$ and there exists a unitary operator $M$ on $\mathcal{H}_A \otimes \mathcal{H}_B$ such that $\Tr_B(M^\dag(\Phi_\psi \otimes \mathcal{I}_B)(\dyad{\Psi}_{AB})M)=\psi$. 
\end{definition}

We say that $\ket{\Psi}_{AB}$ given above is used as quantum one-time pad for faithful one-shot private transfer of $d$-dimensional quantum state. We will drop the modifiers `faithful' and `one-shot' in the following unless it is necessary. 

By definition, both marginal states of the output bipartite state of a PST process should be independent of the input state. Therefore, every PST process stopped before actual transmission is a $((2,2))$-threshold QSS process where $((k,n))$-threshold QSS scheme is a process that encodes an arbitrary quantum state into a $n$-partite quantum state such that only subset of $n$ parties with size larger than $k-1$ can restore the encoded secret quantum state. It was proven that \cite{cleve1999share, gottesman2000theory} only the schemes with $n/2 < k \leq n$ are allowed by the no-cloning theorem and that secret sharing through pure $n$-partite state is possible only for $((k,2k-1))$-threshold QSS schemes.

Especially the impossibility of pure $((2,2))$-threshold QSS schemes (named `masking quantum information' or `quantum masking') was recently rediscovered under the name of the no-masking theorem \cite{modi2018masking}. Subsequently two approaches to circumvent the no-masking theorem have emerged. One is to keep the pureness of the output state and to restrict the set of quantum states to be `masked' (meaning hidden from two local parties) \cite{li2018masking,liang2019complete}. Another is to give up the pureness while keeping the universality, the property of being able to mask any quantum state, by employing the source of randomness \cite{lie2019unconditionally,lie2020randomness}, which is required for any reversible mixed process by the result of Nayak and Sen \cite{nayak2007invertible}.

When it is necessary to distinguish them, we will call the former schemes as unitary masking processes and the latter as randomized masking processes. Note that randomized masking is different from probabilistic masking, which is recently proved to be impossible as well \cite{li2020probabilistic}. Since our focus in this work is on universal processes, when we refer to masking processes without a modifier, it will be the randomized masking processes.

A randomized masking process $\mathcal{T}(\rho):\mathcal{B(H}_A) \to \mathcal{B(H}_A \otimes \mathcal{H}_B)$ should have the following form for any input state $\rho$ \cite{nayak2007invertible},
\begin{equation} \label{eqn:maskdef}
    \mathcal{T}(\rho)= V(\rho \otimes \zeta)V^\dag,
\end{equation}
with a unitary operator $V$ on $\mathcal{H}_A \otimes \mathcal{H}_B$ and some mixed state $\zeta$ acting as a source of randomness, which will be called as the `safe state' of the masking process \cite{lie2019unconditionally}. It is called safe state in the sense that quantum information is securely masked as if it is being stored safely in a virtual safe. The masking process masks quantum information in the sense that both $\Tr_A\mathcal{T}(\rho)$ and $\Tr_B\mathcal{T}(\rho)$ are constant maps regardless of the input state $\rho$.

The equivalence of quantum masking and PST is evident from that both are simply two different expressions of general $((2,2))$-threshold QSS schemes. However, the equivalence of the measures for randomness source of quantum masking and quantum one-time pad of PST can be explicitly shown in the following way. It is known that any purification of $((2,2))$-threshold QSS scheme is a $((2,3))$-threshold QSS scheme \cite{cleve1999share}, and by discarding any one share of a $((2,3))$-threshold QSS scheme one obtains a $((2,2))$-threshold QSS scheme. A purification of the arbitrary masking process $\mathcal{T}(\rho)$ given in Eqn. (\ref{eqn:maskdef}) can be obtained by purifying the safe state $\zeta$, i.e.

\begin{equation}
    (V_{AB} \otimes \mathds{1}_C) \left( \rho_A \otimes \dyad{Z}_{BC}\right) (V_{AB}^\dag \otimes \mathds{1}_C),
\end{equation}
where $\ket{Z}_{BC}$ is a purification of the mixed state $\zeta$. Since this is a $((2,3))$-threshold QSS scheme, by tracing out the system $A$, one gets another $((2,2))$-threshold QSS scheme of the form
\begin{equation}
    (\Phi_\rho \otimes \mathcal{I}_C)(\dyad{Z}_{BC}),
\end{equation}
with $\Phi_\rho(\sigma)\equiv \Tr_A\left[V(\rho\otimes \sigma)V^\dag\right]$ defined for every quantum state $\rho$, which exactly fits the definition of PST. One can easily check that a similar argument holds for the converse case. Therefore the R\'{e}nyi entanglement entropy of the quantum one-time pad $\ket{Z}$ is the same with the R\'{e}nyi entropy of the safe state $\zeta$. Thus, by lower bounding the former, one can also lower bound the latter.

In the following theorem, we give a necessary and sufficient condition that a pure bipartie quantum state $\ket{\Psi}_{AB}$ should satisfy to be able to be used for this task.

\begin{theorem} \label{thm1}
A bipartite state $\ket{\Psi}_{AB}$ can be used as a quantum one-time pad for faithful private transfer of $d-$dimensional quantum state if and only if $S_{\min{}}(\Psi_A)\geq \log_2 d$.
\end{theorem}
\begin{proof}
Suppose that $S_\text{min}(\Psi_A)\geq \log_2 d$. This is equivalent to that the spectrum of $\Psi_A$ is majorized by $(1/d,\dots,1/d,0,\dots)$. If $\ket{\Phi}_{AB}$ is a maximally entangled state on some $d-$dimensional subspaces of $\mathcal{H}_A$ and $\mathcal{H}_B$, by the proof of Nielsen's theorem \cite{nielsen1999conditions} given in Ref. \cite{donald2002uniqueness}, there exists a one-way LOCC superoperator $\Lambda$ on $\mathcal{B}(\mathcal{H}_A \otimes \mathcal{H}_B)$ given in the form of
\begin{equation} \label{eqn:lambda}
    \Lambda(\omega) = \sum_i (K_i \otimes U_i)\omega(K_i^\dag \otimes U_i^\dag)
\end{equation}
where $\{K_i\}$ forms the set of Kraus operators, i.e. $\sum_i K_i^\dag K_i = \mathds{1}_A$ and $U_i$ are unitary operators acting on $\mathcal{H}_B$, such that $\Lambda(\dyad{\Psi}_{AB})=\dyad{\Phi}_{AB}$. We modify this superoperator so that the classical communication from Alice to Bob is suspended and stored in a data storage of Alice, i.e. we extend the superoperator $\Lambda$ to $\tilde{\Lambda} : \mathcal{B}(\mathcal{H}_A\otimes \mathcal{H}_B)\to \mathcal{B}(\mathcal{H}_A\otimes \mathcal{H}_B \otimes \mathds{C}^m)$ for some $m$ (we will name the system $\mathds{C}^m$ as $C$)  given as
\begin{equation} \label{eqn:tilde}
    \tilde{\Lambda}(\omega)=\sum_i (K_i \otimes U_i)\omega(K_i^\dag \otimes U_i^\dag) \otimes \dyad{i}
\end{equation}
and similarly define $\Xi : \mathcal{B}(\mathcal{H}_A)\to \mathcal{B}(\mathcal{H}_A \otimes \mathds{C}^m)$ as
\begin{equation}
    \Xi(\omega) = \sum_i K_i \omega K_i^\dag \otimes \dyad{i}.
\end{equation}
Note that $\Lambda(\dyad{\Psi}_{AB})=\Tr_C \circ  \tilde{\Lambda}({\dyad{\Psi}_{AB}})=\dyad{\Phi}_{AB}$ and $\dyad{\Phi}_{AB}$ is a pure state, therefore $\tilde{\Lambda}({\dyad{\Psi}_{AB}})$ should have a form of $\dyad{\Phi}_{AB} \otimes \sigma_C$, with some $\sigma$ i.e. other systems should be decoupled from a system in a pure state. Then we can see that $\Xi(\Psi_A)=\Phi_A \otimes \sigma_C$ for some mixed state $\sigma$ because $\Xi(\Psi_A)=\Tr_B \circ \tilde{\Lambda}({\dyad{\Psi}_{AB}})$. 

Since $\Phi_A$ is a quantum state with uniform nonzero eigenvalues $1/d$, there exists \cite{lie2019unconditionally, lie2020randomness} a family of secret encoding maps $\{\Theta_\psi\}$ of $d$-dimensional quantum states such that $\Theta_\psi(\Phi_A)$ is constant for all $\psi$. If we let $\{(\Theta_\psi \otimes \mathcal{I}_C)\circ \Xi\}$ as the family of secret encoding maps acting on $A$ of $\ket{\Psi}_{AB}$, we get the wanted result. The secret can be restored by first applying $U_i$ on $B$ conditioned on $C$ followed by discarding the system $C$ and applying the restoring map for $\{\Theta_\psi\}$ on $AB$. The first step transforms $\ket{\Psi}_{AB}$ to $\ket{\Phi}_{AB}$, which is legitimate safe-key state \cite{lie2019unconditionally, lie2020randomness} of $\{\Theta_{\psi}\}$ so that the second step works.

Conversely, suppose that there exists a family of secret encoding maps $\{\Theta_\psi\}$ of $d-$dimensional quantum states acting on $A$ of $\ket{\Psi}_{AB}$. It is equivalent to that there exists a quantum masking process that uses $\Psi_A$ as the safe state \cite{lie2019unconditionally}. Therefore, according to Eqn. (9) of Ref. \cite{lie2020randomness}, every eigenvalue $p_i$ of $\Psi_A$ must not be larger than $1/d$ \bibnote{Since $I(R:A)_{\tau_{RA}}+I(R:B)_{\tau_{RB}}=2\log d$ and $\max\{I(R:A)_{\tau_{RA}},I(R:B)_{\tau_{RB}}\}\leq -\log p_i$, we get $\log d \leq -\log p_i$ for all $i$.}. Since it is equivalent to $S_\text{min}(\Psi_A)\geq \log_2 d$, the wanted result is obtained.
\end{proof}
This result generalizes the result of Ref. \cite{gour2004faithful} since the quantum teleportation is a special case of PST. Note that any eavesdropper of the classical communication in a teleportation protocol without sharing initial entanglement cannot gain any information of the teleported quantum state. Also, Theorem \ref{thm1} shows that the faithful teleportation protocol given in Ref. \cite{gour2004faithful} is not only an optimal teleportation protocol, but also an optimal PST protocol in the sense that the protocol consumes the 
minimal amount of entangled state without any leftover entanglement and does not require quantum channel between Alice and Bob for the secret recovery.

If we define the one-shot PST power of $\ket{\Psi}_{AB}$ as the maximal size of transferable quantum state by using state $\ket{\Psi}_{AB}$ counted in qubits as $P_p(\ket{\Psi}_{AB}) \equiv \log_2 \lfloor 2^{S_{\min{}}(\Psi_A)}\rfloor$, we have the following result.
\begin{corollary}
    The one-shot PST power for bipartite states is superadditive, i.e. $P_p(\ket{\Psi}\otimes\ket{\Phi})\geq P_p(\ket{\Psi})+P_p(\ket{\Phi})$.
\end{corollary}

\section{Quantum Masking}

From the equivalence and duality of quantum masking and PST, we can similarly define the masking power of a quantum state $\sigma$ as $P_m(\sigma)\equiv\lfloor2^{S_{\min{}}(\sigma)}\rfloor$. We will say that a quantum state $
\sigma$ can mask $d$-dimensional quantum information when it is used as the safe state of a $d$-dimensional randomized quantum masking process. The main result implies the following.

\begin{corollary} \label{coro:mask}
    A quantum state $\sigma$ can mask $d$-dimensional quantum information if and only if $S_{\min{}}(\sigma)\geq \log_2 d.$ Moreover, the masking power quantum state is superadditive, i.e. $P_m(\sigma_1 \otimes \sigma_2) \geq P_m(\sigma_1) + P_m(\sigma_2).$
\end{corollary}

Note that, by appropriately merging unauthorized sets of an arbitrary $((k,n))$-threshold QSS scheme, one can construct a $((2,2))$-threshold scheme \cite{lie2019unconditionally, lie2020randomness}, i.e. quantum masking, therefore Corollary \ref{coro:mask} applies to an arbitrary unauthorized set of $d$-dimensional QSS; i.e. any unauthorized party's marginal state should have the min-entropy larger than or equal to $\log_2 d$.

This result implies that having large von Neumann entropy alone is not enough for masking quantum information. For example, a rank 3 quantum state with the spectrum of $(0.7730,0.1135,0.1135)$ has 1 bit of von Neumann entropy, but since its min-entropy is 0.3716 bits, it cannot mask a qubit of quantum information. On the other hand, a state with the spectrum of $(1/2,1/4,1/8,\dots,1/2^n,1/2^n)$ which has 1 bit of min-entropy can mask a qubit of quantum information even though its randomness is highly non-uniform. The consequences of Theorem \ref{thm1} is not limited to quantum information processing tasks, but also has implications for the field of state transition.

\section{State transition}

For any initial state $\rho$ and the final state $\rho'$, implementing a quantum channel (`transition') $\mathcal{E}$ such that $\mathcal{E}(\rho)=\rho'$ by utilizing randomness has been studied \cite{gour2015resource, boes2018catalytic, scharlau2018quantum}. Every state transition process can be realized if dephasing map can be realized \cite{boes2018catalytic, scharlau2018quantum} due to the Schur-Horn theorem \cite{horn1954doubly}, since every state transition can be decomposed into the initial unitary evolution followed by a dephasing map and the final unitary evolution. Therefore, realizing dephasing maps is the essential part of implementing state transition. The necessary and sufficient condition for the source of randomness (SOR) of dephasing channel can be also obtained from Theorem \ref{thm1}.  Here, we will use a slightly different definition of (catalytic) dephasing map using quantum randomness to encompass the usage of non-perfect SOR \cite{lie2020only}. We will say the map $\mathcal{E}$ dephases with respect to a certain basis $\{\ket{i}\}$ using SOR $\sigma$ with a leftover SOR $\eta$ if there exists an isometry operator (a unitary operator that embeds a smaller Hilbert space into a larger one) $U$ acting on $AB$ such that for any $d$-dimensional quantum state $\rho$,
\begin{equation}
        \mathcal{E}(\rho)=\Tr_B\left[U(\rho \otimes \sigma) U^\dag\right]= \sum_i \bra{i}\rho\ket{i}\dyad{i} \otimes \eta,
\end{equation}
and there exists some quantum state $\tau$ on $B$ so that the complement channel of $\mathcal{E}$ has the form of
\begin{equation}
    \tilde{\mathcal{E}}(\rho) = \Tr_A\left[U(\rho \otimes \sigma) U^\dag\right] = \tau,
\end{equation}
regardless of $\rho$. We will say the use of SOR $\sigma$ is \textit{catalytic} when $\tau$ can also be used for some $d$-dimensional dephasing map with the same property, i.e. it is an infinitely recyclable SOR. One can see that this definition is recursive. We will also say that $\sigma$ (catalytically) dephases $d$-dimensional quantum states for the same situation. The leftover SOR $\eta$ does not cause problems since it can always be stored or discarded independently of the dephasing process 
itself, since it is in a product state with the dephased state of the input $\rho$.  If the second requirement was not imposed, then SOR is not needed at all since a simple CNOT gate can dephase with a pure ancillary system. See Ref. \cite{lie2020only} for a more detailed discussion on this generalized setting.

In Ref. \cite{boes2018catalytic}, only maximally mixed SORs were considered when the minimal randomness bound was derived and only approximate dephasing was considered for potentially non-uniform SORs, but since it is well-known that non-uniform randomness sources can cause security issues \cite{mcinnes1990impossibility, dodis2004possibility}, it is necessary to analyze the power of non-uniform SORs. In the following theorem, we give a necessary and sufficient condition of when a SOR can be used to \textit{exactly} dephase an arbitrary input state.

\begin{theorem}
    A quantum state $\sigma$ can dephase $d^2$-dimensional quantum states catalytically if and only if $S_{\min{}}(\sigma) \geq \log_2 d$.
\end{theorem}
\begin{proof}
    Suppose that $\sigma$ can dephase $d^2$-dimensional quantum states. Then, by using two copies of $\sigma$, i.e. $\sigma \otimes \sigma$, one can mask any $d^2$-dimensional quantum state $\rho$ by dephasing it into two mutually unbiased bases. For example, one can use one $\sigma$ to dephase $\rho$ and applying the ($d^2$-dimensional) discrete Fourier transform gate (defined as $\sum_{n,m=1}^{d^2}\exp(i{2\pi nm}/{d^2}) \dyad{n}{m}$) to the output. Then, by using the other $\sigma$, one can dephase the output state with respect to the same basis. The final output state is the $d^2$-dimensional maximally mixed state for every input state $\rho$ (with some leftover SOR in a product state ). Since the SOR $\sigma \otimes \sigma$ is also transformed into a quantum state that is independent of $\rho$, the whole process is a randomized masking process. Therefore, by Corollary \ref{coro:mask}, $S_{\min{}}(\sigma \otimes \sigma) \geq 2\log_2 d$. By the additivity of the min-entropy, we get $S_{\min{}}(\sigma) \geq \log_2 d$.
    
    Conversely, assume that $S_{\min{}}(\sigma) \geq \log_2 d$. If we pick a purification $\ket{\Sigma}$ of $\sigma$, then one can replace $\ket{\Psi}$ in the proof of Theorem \ref{thm1} with $\ket{\Sigma}$ since it majorizes a $d$-dimensional maximally entangled state $\ket{\Phi}$. Therefore we will use the corresponding Kraus operators $\{K_i\}$ defined in the same way with Eqn. (\ref{eqn:tilde}). Then, we apply the isometry operator $\sum_i \ket{i}_{A'} \otimes \ket{i}_{B'} \otimes K_i$ ($A'$ and $B'$ belongs to Alice and Bob respectively, and $K_i$ acts on $B$, the same system with that of $\sigma$) to $\sigma$. Then, for both Alice and Bob, who has no access to $B'$ and $A'$ respectively, the system $B$ is uncorrelated to their primed systems ($A'$ and $B'$) and the system $B$ is in the rank $d$ uniformly mixed state $\Phi_B$. (See the proof of Theorem \ref{thm1}.) Then, by applying the optimal $d^2$-dimensional dephasing unitary given in Ref. \cite{boes2018catalytic}, which uses the state $\Phi_B$ as a catalyst, we can realize a $d^2$-dimensional dephasing map with the leftover SOR $\kappa \equiv \sum_i \Tr\left[K_i\sigma K_i^\dag\right]\dyad{i}$, and the SOR used in this process is transformed into $\Phi_B \otimes \kappa_{B'}$. Since $S_{\min{}}(\Phi_B \otimes \kappa_{B'}) \geq S_{\min{}}(\Phi_B) = \log_2 d$, this SOR can be used again for another $d^2$-dimensional dephasing map, therefore $\sigma$ was used catalytically in this process.
\end{proof}

The catalyst used in this process has transformed from $\sigma$ to $\Phi_B\otimes \kappa_{B'}$. We remark that, however, the catalyst's min-entropy never decreases during the process. This follows from that, for all $i$,
\begin{equation}
    2^{-S_{\min{}}(\Phi_B \otimes \kappa_{B'})}= \max_i \frac{1}{d}\Tr\left[K_i \sigma K_i^\dag \right] \leq 2^{-S_{\min{}}(\sigma)},
\end{equation}
where the first equality follows from that $2^{-S_{\min{}}(\Phi_B \otimes \kappa_{B'})}$ is the largest eigenvalue of $\Phi_B \otimes \kappa_{B'}$, and the second inequality follows from that $\sigma \leq 2^{-S_{\min{}}(\sigma)} \mathds{1}$ and $K_i^\dag K_i \leq \Pi_{\text{supp}(K_i)}$, where $\Pi_{\text{supp}(K_i)}$ is the orthogonal projector onto the support of $K_i$ (the orthogonal complement of the kernel of $K_i$), and that $\Tr\Pi_{\text{supp}(K_i)}=d$. Thus no consumption of randomness in terms of min-entropy happens in the process.

One can even recover full catalycity, if decoherence is allowed, by applying the following controlled unitary to $\Phi_B\otimes \kappa_{B'}$ given as $\sum_i U_i^\dag \otimes \dyad{i}$ with  the unitaries $\{U_i\}$ of the proof of Theorem \ref{thm1} and discarding the latter system. This is because $\sum_i \Tr\left[K_i \sigma K_i^\dag\right]U_i^\dag\Phi_BU_i=\sigma$. Therefore, as it was assumed in \cite{boes2019neumann}, if the $\kappa_{B'}$ undergoes decoherence and is dephased with respect to a basis unbiased from its eigenbasis, we can see that the catalyst goes back to its original form, $\sigma$, with some uncorrelated leftover SOR, the state $\kappa$ after being dephased.

On the other hand, if one wants to remove the leftover SOR $\kappa$ completely to follow the conventional formalism of catalytic quantum randomness \cite{boes2018catalytic}, note that a simple projective measurement on $\kappa$ will collapse it into a pure state and leave the system $A$ in the dephased state without the leftover SOR (up to local unitary) and the catalyst in the state $\Phi_B$, which is the standard form of catalyst \cite{boes2018catalytic}, \textit{deterministically}, regardless of the measurement outcome. This is because  both $\kappa_{A'}$ and $\kappa_{B'}$ are completely decoupled with the system $A$ and $B$, respectively. In this sense, one can say that the initial catalyst $\sigma$ was actually \textit{precatalyst}, a compound that is converted into catalyst during the chemical reaction, since it produces the true catalyst $\Phi_B$ in the course of interaction.

\section{Conclusion}

Weak randomness sources with non-uniform probability distribution often causes nonzero probability of failure \cite{von195113}, sometimes to an irremediable extent \cite{mcinnes1990impossibility, dodis2004possibility} when their distribution is not fixed. We showed, however, for a fixed source of randomness or entanglement, non-perfect resources can yield \textit{deterministic} security when their min-entropy is larger than the size of the secret they are hiding.

Aside from the results on quantum information processing tasks, we showed the power of non-uniform randomness sources as a catalyst, which could be interpreted as a thermal machine \cite{muller2018correlating} in quantum thermodynamics context, when it comes to state transition. By utilizing the Nielsen theorem \cite{nielsen1999conditions}, which was initially applied to entanglement extraction, we showed that non-uniform randomness sources can also serve as catalyst and provided a necessary and sufficient condition for when it is possible. Dephasing map can be also be understood as a quantum masking process of an observable. Our result shows that the randomness cost of masking one observable is the half of the cost of masking the whole quantum information in a system. This result solidifies the intuition that the information of one observable amounts to the half of the quantum information in the same system.

Anticipated future work is to extend this result for arbitrary mixed quantum one-time pads. Unlike the results for safe states and SOR for dephasing maps, the result given in this work is not fully general for quantum one-time pads since we only considered pure quantum one-time pads. 

\begin{acknowledgments} 
This work was supported by the National Research Foundation of Korea (NRF) through grants funded by the the Ministry of Science and ICT (Grants No. NRF-2019M3E4A1080074 and No. NRF-2020R1A2C1008609).
\end{acknowledgments} 

\bibliography{main}

\end{document}